\documentclass[10pt,conference,twocolumn,letterpaper]{IEEEtran}
\makeatletter
\IEEEoverridecommandlockouts
\usepackage{graphicx,amsmath,amssymb,url,verbatim}
\usepackage[tight]{subfigure}

\newcommand{\gv}{{\bf g}}
\newcommand{\ov}{{\bf o}}

\newcommand{\FF}{\mathbb{F}}

\newcommand{\Bc}{\mathcal{B}}

\newcommand{\Lc}{\mathcal{L}}

\newcommand{\soc}{\mathrm{S}}
\newcommand{\dc}{\mathrm{DC}}
\newcommand{\xin}{x_{\rm{in}}}
\newcommand{\xcor}{x_{\rm{coord}}}
\newcommand{\xout}{x_{\rm{out}}}
\newcommand{\lmax}{l^{max}}
\newcommand{\bmax}{b^{max}}

\newtheorem{ex}{Example}
\newtheorem{prop}{Proposition}

\title{Byzantine Fault Tolerance of Regenerating Codes\thanks{F. Oggier's research for this work has been supported by the Singapore National Research Foundation under Research Grant NRF-RF2009-07. A. Datta's research for this work has been supported in part by A*Star TSRP grant number 102 158 0038 on pCloud project and in part by NTU AcRF Tier-1 grant number RG 29/09 on CrowdStore project. The authors will also like to thank Nicolas Le Scouarnec for his advices to carry out the numerical optimizations.}}
\begin{document}
\author{
\IEEEauthorblockN{Fr\'ed\'erique Oggier}
\IEEEauthorblockA{Division of Mathematical Sciences \\
School of Physical and Mathematical Sciences\\
Nanyang Technological University, Singapore\\
Email: frederique@ntu.edu.sg}
\and
\IEEEauthorblockN{Anwitaman Datta}
\IEEEauthorblockA{Division of Computer Science\\
School of Computer Engineering\\
Nanyang Technological University, Singapore\\
Email: anwitaman@ntu.edu.sg}
}
\maketitle

\begin{abstract}
Recent years have witnessed a slew of coding techniques custom designed for networked storage systems. Network coding inspired regenerating codes are the most prolifically studied among these new age storage centric codes. A lot of effort has been invested in understanding the fundamental achievable trade-offs of storage and bandwidth usage to maintain redundancy in presence of different models of failures, showcasing the efficacy of regenerating codes with respect to traditional erasure coding techniques. For practical usability in open and adversarial environments, as is typical in peer-to-peer systems, we need however not only resilience against erasures, but also from (adversarial) errors. In this paper, we study the resilience of generalized regenerating codes (supporting multi-repairs, using collaboration among newcomers) in the presence of two classes of Byzantine nodes, relatively benign selfish (non-cooperating) nodes, as well as under more active, malicious polluting nodes. We give upper bounds on the resilience capacity of regenerating codes, and show that the advantages of collaborative repair can turn to be detrimental in the presence of Byzantine nodes. We further exhibit that system mechanisms can be combined with regenerating codes to mitigate the effect of rogue nodes.
\end{abstract}
 \textbf{Keywords:} distributed storage, regenerating codes, Byzantine faults, pollution, resilience
%
%

\section{Introduction}
Redundancy is essential for reliably storing data. This basic principle has been adhered in designing diverse storage solutions such as CDs and DVDs, RAID systems as well as, more recently - networked distributed storage systems. Such redundancy may be achieved by replicating the data, or applying coding based techniques. Coding based techniques incur much less storage overhead with respect to replication based technique in order to achieve equivalent resilience (fault-tolerance). Thus, coding based redundancy is often preferred for efficiently storing large amount of data.

In networked storage systems, which may be as diverse as peer-to-peer (P2P) storage systems or data centers, redundant data is distributed across multiple storage devices. When some of these devices become unavailable - be it due to failure or (permanent) churn, redundancy needs to be replenished, otherwise, over time, the system will lose the stored data. If replication based redundancy is used, a new replica is created by copying data from existing replica(s). When using coding based techniques, each storage node typically possesses a small (w.r.to the size of the original data being stored) amount of the data, that we will call an \emph{encoded block}. Since the data can be recovered by contacting a fraction of the storage nodes, redundancy can be replenished in the same way: first reconstruct the whole data, re-encode it, and re-distribute the encoded blocks.

This is the case when using traditional \emph{erasure codes} (EC) such as Reed-Solomon codes \cite{ReedSolomon}. In order to replenish lost redundancy, data equivalent in volume to the complete object needs to be transferred (or stored at one node a priori), in order to recreate even a single encoded block. To improve on such a naive approach, network coding based coding \cite{DGWK-infocom07} was proposed to recreate one new encoded block by transferring much less data, upto possibly equivalent volume of data to only as is to be recreated. This new family of codes is called \emph{regenerating codes} \cite{WDR} - and the strategy may be applied on the original data itself, or on top of erasure encoding. Two different types of works have emerged on regenerating codes: those which establish the theoretical feasibility of such bandwidth efficient redundancy replenishment through min-cut bounds (such as \cite{WDR}, or \cite{Shum-ICC} for more general bounds), and those which instead try to provide various coding strategies to do so in practice.

The current regeneration code related literature mostly (but for \cite{RDC-2010} and \cite{PER-arxiv} that we will discuss later on) assumes a friendly environment, where all live nodes are well behaved. In open environments, particularly P2P environments, one should make such an assumption at his own peril.

We note that erasure codes such as Reed-Solomon codes are resilient against not only `erasures' but are also capable of dealing with `errors'. In contrast, while regenerating codes inherit the advantages of network coding such as bandwidth efficiency, they also likewise suffer from the same vulnerabilities of network coding. One of the most critical issues which intrinsically affect
network coding is the family of pollution attacks. The idea behind
network coding is to allow any intermediate node in the network
to forward linear combinations of its incoming packets to its neighbors,
which when done cleverly and diligently, results in throughput gain. However, it also means that one bogus
packet can corrupt several other packets downstream, and thus spread over and contaminate
a large portion of the network. Such attacks are not possible in a classical routing scenario.

The same problem of pollution attack can be directly translated in the context of coding for
distributed storage based on network coding, in particular in the case of
regenerating codes. In this paper, we study if and how well regenerating codes may tolerate Byzantine nodes. We identify the cardinal Byzantine attacks possible during the regeneration process. Specifically, we look at the following families of Byzantine nodes: \begin{itemize}
\item \emph{Selfish (non-cooperating) nodes}: Nodes may not actively attack the network, however they may prioritize their own interests, and might just decline to cooperate during the regeneration process, that is, refuse to provide the data that is requested from them to carry out regeneration. In absence of the contribution from such selfish or non-cooperating nodes, a regeneration protocol designed assuming their contribution will fail to carry out the regeneration task anymore.
\item \emph{Polluters}: Nodes may try to disrupt the regeneration process actively, by deliberately sending wrong data. Such active attack is particularly detrimental while using regenerating codes, since it would affect future regeneration processes where a victim participates and continues to further spread the pollution unconsciously and unintentionally.
\end{itemize}

The main contributions of this work are as follows. (i) We determine bounds on the resilience capacity of regenerating codes, taking into account the above mentioned  adversarial behaviors. (ii) Our analysis reveals that though collaboration in regeneration can be beneficial in terms of bandwidth and storage costs, the penalty in presence of Byzantine nodes is also substantially larger. There is a blowback effect, in that, collaboration may not only be useless under Byzantine attacks, but can in fact be detrimental, such that one would be better off by avoiding collaboration. (iii) Finally, we outline how this effect can also be easily mitigated in practice using some additional information and extrinsic mechanisms.


%
%

\section{Regenerating codes in a nutshell}
\label{sec:RGC}

Consider an object of size $B$ to be stored in a network with $n$ storage nodes, a source $\soc$ which has adequate bandwidth to upload data
over the network to these nodes, and a data collector $\dc$ which should be able to retrieve a given stored
object by accessing data from any arbitrary choice of $k$ out of the $n$ nodes. Thus to say, such a storage network stores the object redundantly, and can tolerate up to $n-k$ failures without affecting the object's availability. For instance, erasure codes may be used to encode the object and achieve such redundancy.

Over time, some of the storage nodes may go offline (or crash), and if the redundancy is not restored then  the system's fault tolerance will reduce, leading to, in the worst case, eventual loss of the stored object. Thus, mechanisms are needed to repair or regenerate the lost redundancy. Naive solutions include keeping a full copy of the object somewhere, which can be used to recreate the lost data at any node. Alternatively, if no such full copy is available, then one can download adequate, i.e., $k$ encoded data blocks, and use these to regenerate the lost encoded data blocks. These naive solutions are sub-optimal in terms of efficient use of storage space and bandwidth for regeneration respectively, and have in the recent years prompted the exploration of better solutions - such as (chronologically) Pyramid codes \cite{pyramid}, Regenerating codes \cite{WDR}, Hierarchical codes \cite{hierarchical} and Self-repairing codes \cite{SRC} to name some of the most prominent ones. We next summarize some key results related to regenerating codes, since this paper studies their Byzantine fault tolerance.

Suppose that each node has a storage capacity of $\alpha$, i.e., the size of the encoded data block stored at a node is of the size $\alpha$. When one data block needs to be regenerated, a new node contacts $d$ ($k \leq d$) other existing nodes, and downloads $\beta$ amount of data from each of the contacted nodes (referred to as the bandwidth capacity of the connections between any node pair)\footnote{Note that, in contrast to conventional techniques which download the whole encoded data block, only a smaller $\beta/\alpha$ fraction of data from each contacted node is being transferred.}. By considering an information flow from the source to the data collector, a trade-off between the nodes' storage capacity and bandwidth can be computed \cite{WDR}, through a min-cut bound.
\begin{prop}\cite{WDR}
\label{prop:rgc}
A min-cut bound of an information flow  between the source and a data collector is
\[
\mbox{mincut}(\soc,\dc) \geq \sum_{i=0}^{k-1}\min\{\alpha,(d-i)\beta\}.
\]
\end{prop}
Note that such a min-cut bound determines achievability - without necessarily stating any specific way to actually do so.
Furthermore, it is required that
\[
\sum_{i=0}^{k-1}\min\{\alpha,(d-i)\beta\} \geq B
\]
for regeneration to be possible.
Two sub-families of regenerating codes have consequently emerged \cite{WD-isit09} - coined as \emph{functional}, respectively \emph{exact}, to provide actual coding strategies.
Functional repair strategies rely on random network coding arguments, and while they regenerate lost redundancy, the data stored by new nodes is not `bit-by-bit' identical to the encoded block that previously existed: it is enough that it allows the retrieval of the stored data. In contrast, exact repair leads to regeneration of bit-by-bit identical encoded block as was lost. Exact regeneration is preferable since it translates to simplicity in system design and management. A more detailed
comparison between exact and functional repair can be found in \cite{biersackRGCstudy}.

The original bound reported in Proposition \ref{prop:rgc} was derived assuming that only one encoded data block for a single node is being regenerated. However, this is not a realistic assumption to build practical networked storage systems. In highly dynamic scenarios, which is typical in peer-to-peer environments, but also may happen in more static (data-center like) environments due to correlated failures, it may be necessary to regenerate data for multiple nodes. Naive strategies would include regenerations sequentially, or in parallel, but independently of each other.

In \cite{multiRGC}, the above framework has been extended for multiple new nodes to carry out regeneration by not only downloading data from (old) live nodes, but also by additionally collaborating among each other under some specific settings. A more generalized result is provided in \cite{Shum-ICC} (and also, independently in \cite{beyondRGC}).

The regeneration process is carried out in two phases, a download
phase during which a batch of $t$ newcomers download data from any $d$
live nodes each, and a collaborative phase, where each newcomer
shares some of its data to help the $t-1$ other new nodes. Such a two
phase regeneration involving collaboration among new nodes can lead
to reduction in the overall bandwidth usage for the regenerations.

Under such a setting, a more general min-cut bound is derived. In
the following, $\beta'$ represents the bandwidth during the
collaborative phase, i.e., each new node sends (and also receives)
$\beta'$ data to (from) each other new node. Consider that the data collector contacts $k$ nodes for reconstructing the data, such that the contacted nodes can be arranged in $g$ groups of sizes $u_i$ where $u_0+\ldots+u_{g-1}=k$, where each such group represents a generation of $t$ nodes which had joined the system together and carried out the regeneration collaboratively.

\begin{prop}\label{prop:collabrgc}\cite{beyondRGC,Shum-ICC}
A min-cut bound of an information flow between the source and a data collector is
\[
\mbox{mincut}(\soc,\dc)\geq \sum_{i=0}^{g-1} u_i\min\{\alpha,
(d-\sum_{j=0}^{i-1}u_j)\beta+(t-u_i)\beta'\}
\]
where $k=\sum_{i=0}^{g-1}u_i$ with $1\leq u_i \leq t$,
\end{prop}
and as above, we need
\[
\sum_{i=0}^{g-1} u_i\min\{\alpha,
(d-\sum_{j=0}^{i-1}u_j)\beta+(t-u_i)\beta'\}\geq B
\]
for regeneration to be possible.
When $t=1$, we get that $u_i=1$, thus $g=k$ and the more general bound matches the one given in
Proposition \ref{prop:rgc}.

As pointed out in \cite{beyondRGC}, two extreme cases can be identified. First, if there is no contribution in
$\beta'$, then the highest contribution comes from $\beta$, that is $u_i=t$ and $g=k/t$, and the min-cut bound becomes
\begin{equation}\label{eq:simpmincut1}
\mbox{mincut}(\soc,\dc)\geq
\sum_{i=0}^{k/t-1}t \min\{\alpha,(d-it)\beta\}.
\end{equation}
Conversely, the highest contribution from $\beta'$ comes when
$\beta$ is minimized, which occurs when $u_i=1$ for all $i$ and $g=k$. Then the min-cut bound simplifies to
\begin{equation}\label{eq:simpmincut2}
\mbox{mincut}(\soc,\dc)\geq
\sum_{i=0}^{k-1} \min\{\alpha,
(d-i)\beta+(t-1)\beta'\}.
\end{equation}

The minimum possible amount of data that can be stored at a node is $B/k$, since the data collector must be able to
retrieve the object out of any $k$ nodes.
Codes using the lowest amount of storage $\alpha=B/k$ are said to satisfy the {\em minimum storage regeneration (MSR) point},
and using (\ref{eq:simpmincut1}) and (\ref{eq:simpmincut2}) are shown to be characterized by \cite{beyondRGC}
\begin{equation}\label{eq:mscr}
\alpha=\frac{B}{k},~\beta=\beta'=\frac{B}{k}\frac{1}{d-k+t}
\end{equation}
while codes requiring the minimum bandwidth for regeneration similarly satisfy
\begin{equation}\label{eq:mbcr-alpha}
\alpha=\frac{B}{k}\frac{2d+t-1}{2d-k+t}
\end{equation}
and
\begin{equation}\label{eq:mbcr-beta}
\beta=\frac{B}{k}\frac{2}{2d-k+t},~\beta'=\frac{B}{k}\frac{1}{2d-k+t},
\end{equation}
a point called the {\em minimum bandwidth regeneration (MBR) point}.

\begin{figure}[htbp]
  \begin{center}
    \includegraphics[scale=0.6]{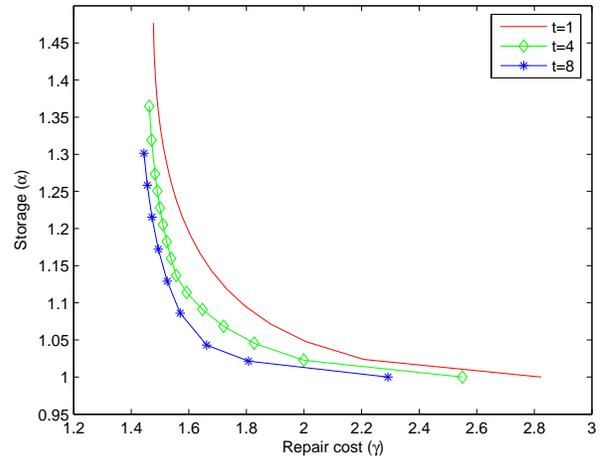}
    \caption{The storage bandwidth (per repair) trade-off curve using regenerating codes with collaboration for $t =$ 1, 4, 8. This plot (and all others plots in this paper) has been generated using linear non-convex optimization numerically. The values have been normalized by $B/k$.}
  \label{fig:CollabRGC}
  \end{center}
    \vspace{-0.5cm}
\end{figure}

The benefit of collaborative regenerating codes with respect to standard regenerating codes (that is, with no collaboration phase) is illustrated in Fig.~\ref{fig:CollabRGC}, where we set $d=48$ and $k=32$. Trade-off curves between the storage cost $\alpha$ on the \emph{y-axis} against the bandwidth cost per repair on the \emph{x-axis}, determined in (\ref{eq:gammaCollabRGC}) and denoted as $\gamma$ are shown for different scenarios. For collaborative regenerating codes, the total bandwidth for one node to be repaired is the data downloaded from live nodes, that is $\beta$ from $d$ nodes, and the data exchanged among newcomer nodes during collaboration, which is $\beta'$ from $t-1$ nodes, for a total of
\begin{equation}\label{eq:gammaCollabRGC}
\gamma=d\beta+(t-1)\beta'.
\end{equation}
If no collaboration is done, then $t=1$ and $\gamma=d\beta$. The trade-off curve for $t=1$ in Fig.~\ref{fig:CollabRGC} thus
corresponds to standard (independent) regenerations. Larger value of $t$, implying multiple repairs being carried out collaboratively, allows the storage system to operate using both lower storage and bandwidth costs.

Though several works discussed min-cut bounds for collaborative regeneration codes, we are aware of only one family of collaborative regenerating codes \cite{Shum-ICC}, which provides exact repair for $d=k$ at MSR point. It is noted in \cite{beyondRGC} that for $d=k$, the repair cost is the same as for erasure correcting codes
using delayed repair.

%
%

\section{Byzantine faults model}

The regeneration process can be dramatically affected if some of the live nodes behave in a Byzantine manner, that is, act in a manner different than as expected by the regeneration process. So far, and to the best of our knowledge, \cite{PER-arxiv} is the only work looking at security issues related to regenerating codes. Besides considering a passive adversary who eavesdrops, it also looks at malicious behaviors affecting data integrity at nodes during the regeneration process, but all the considered scenarios assume a single regeneration at a time, rather than the more general problem of multiple simultaneous regenerations. This naturally excludes the complications arising due to the collaboration phase, where a single Byzantine node can potentially contaminate all the other regenerating nodes simultaneously.

In this paper, we consider two types of Byzantine adversaries. A relatively benign form of faulty behavior is when a live node does not provide any data for the regeneration process. We will refer to such Byzantine nodes as \emph{selfish nodes}. Note that we distinguish a selfish node from an unavailable (offline) node in that a selfish node is expected to continue to respond to a data collector trying to recreate the object. If a node refuses to help for both regeneration and also data access, then it can be treated analogously as any other offline node. Such a selfish behavior may arise due to various reasons: the node may be overloaded with other tasks, or there may be temporary problems in the communication link - so that the node can not respond in a timely manner to meaningfully contribute to the regeneration process. No such time-bounded response is assumed for data reconstruction by a data collector. Alternatively, a node participating in a peer-to-peer back-up system may be comfortable with responding to data access requests which are relatively infrequent, and hence less taxing on its bandwidth resources, than regeneration process which could be frequent due to system churn, prompting the node to act selfishly for the regeneration process.

A more malign faulty behavior is when wrong data is sent by a node. Such a behavior even by a single node, if unchecked, may corrupt many nodes downstream. We will refer to such nodes as \emph{polluting nodes}. Rapid propagation of pollution is an inherent and general weakness of network coding, on which rely regenerating codes, making the system extremely vulnerable in the presence of even one or very few polluting nodes.

We note that, for collaborative regeneration, the Byzantine nodes may be among the originally online nodes when the regeneration process is initiated; or among the newly joining nodes, i.e., during the collaboration phase; or a mix of both. Clearly, the amount of data that can be stored reliably and needs to be transferred during regeneration will change under these adversarial constraints, and in particular, so will the trade-off between the storage $\alpha$ and the bandwidth $\beta,\beta'$ as described by the min-cut bounds in Propositions 1 \& 2.

In the spirit of \cite{PER-arxiv}, we consider the resiliency capacity of the distributed storage system as the maximum amount of data that can be stored reliably over the network in the presence of malign nodes, and made available to a legitimate data collector. More precisely, we will focus on the resiliency capacity $C_{r,s}(\alpha,\beta,\beta')$ in the presence of selfish nodes, and $C_{r,p}(\alpha,\beta,\beta')$ when polluting nodes are active.

%
%
%

\section{Min-cut bounds under Byzantine failures}

We will analyze how the storage bandwidth trade-off given in Proposition \ref{prop:collabrgc} is affected in presence of the various Byzantine nodes. We study the general case of regenerating codes, which studies multiple simultaneous regenerations, and with collaboration among the new nodes.

We determine upper-bounds, which means that it is not possible to do any better than the constraints of the corresponding bounds. Note that this is in contrast to the Propositions \ref{prop:rgc} \& \ref{prop:collabrgc}, which determined achievability, though both bounds are derived through min-cut computations.
Since we derive upper bounds here, we can make simplifying (optimistic) assumptions, implying that, under more realistic assumptions and complicated derivations, it may be possible to determine tighter bounds.

\begin{figure}[htbp]
  \vspace{-0.4cm}
    \includegraphics[scale=0.36]{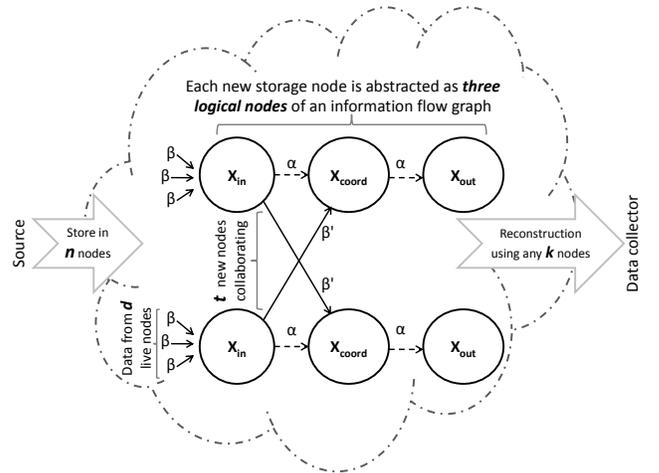}
    \caption{An abstract information flow graph model for the coordinated regeneration process.}
  \label{fig:RGCCloud}
  \vspace{-0.4cm}
\end{figure}

For determining a min-cut, we consider an information flow graph and
use the same abstraction as in \cite{beyondRGC}, which is illustrated in Fig.~
\ref{fig:RGCCloud}. Each new storage node is modeled using three
logical nodes in an information flow graph connecting the source to the
data collector, namely $\xin$, $\xcor$
and $\xout$. It is assumed that $t$ such new nodes carry out the
regeneration in a collaborative manner. $\xin$ represents the
aggregation of information by a new node from $d$ of the existing
live nodes, collecting $\beta$ data from each such contacted live
nodes. In the next (collaborative) phase, each new node provides
(and also obtains) $\beta'$ data from each of the other new nodes.
This collected data is then processed at individual nodes, and
finally they retain (store) $\alpha$ amount of data each. Thus to
each node corresponds a triple
$\xin\rightarrow\xcor\rightarrow\xout$ where both edges
$\xin\rightarrow\xcor$ and $\xcor\rightarrow\xout$ have a
capacity of $\alpha$. We will later (in Example \ref{ex:icc},
Section \ref{sec:ICCExactRGC}) elaborate a concrete example of
multiple regenerations with coordination.

\subsection{Effect of selfish nodes}
In the following we assume that the number of selfish nodes among the live (old) nodes is given by $\Lc_0$ in any generation, and $l_i\leq\lmax$ is the number of selfish nodes among the $i$th group of new comers, for some upper bound $\lmax$.
The total number of selfish nodes participating in the collaborative phase of regeneration over $g$ generations is $\Lc=\sum_{i=0}^{g-1}l_i$.

\begin{prop}
The resiliency capacity $C_{r,s}(\alpha,\beta,\beta')$ in the presence of selfish nodes is upper bounded by
\[
\begin{array}{c}
C_{r,s}(\alpha,\beta,\beta')\leq\\
\!\!\!\sum_{i=0}^{g-1} u_i\min\{\alpha,(d-\Lc_0-\sum_{j=0}^{i-1}u_j)\beta+(t-l_i-u_i)\beta' \}.
\end{array}
\]
\end{prop}
\label{prop:selfishcollabRGC}
\begin{proof}
Consider a cut of the network, between the set $U$ which contains the source $\soc$,
and its complementary set $\bar{U}$ which contains the data collector $\dc$.
The information flow goes from the source to the data collector, through
$\xin \rightarrow \xcor \rightarrow \xout$, where both edges are
assumed to have capacity $\alpha$.
Let $u_0$ be the number of new comers contacted by the data collector in the first group of $t$
new comers, with $m$ of them in $U$,
and $u_0-m$ of the others in $\bar{U}$. Take a first node,
if it belongs to $U$, then it contributes to $\alpha$ (if either $\xcor\in U$ or $\xcor\in \bar{U}$) to the cut,
thus the $m$ nodes in $U$ contribute to a total of $m\alpha$ to the cut.

Consider now the $u_0-m$ nodes in $\bar{U}$. There are two contributions to the cut,
coming from either $\xin$ or $\xcor$. The $\xin$ part downloads from live nodes, of which, there are $\Lc_0$ selfish nodes. In an adversarial scenario, the first $\Lc_0$ nodes
contacted may all be selfish, and as a result, the contribution to the cut would be
$(d-\Lc_0)\beta$. Now for $\xcor$, it contacts $t-1$ other new comers, $u_0-m$ could already be
in $\bar{U}$ (including itself), and $l_0$ could be selfish, thus the cut is increased of $(t-(u_0-m)-l_0)\beta'$, for a total of
\begin{eqnarray*}
c_0(m) & \geq & m\alpha+(u_0-m)[(d-\Lc_0)\beta+ \\
       &      & (t-u_0+m-l_0)\beta']\\
       & \geq & u_0 \min\{\alpha,(d-\Lc_0)\beta+(t-l_0-u_1)\beta'\}
\end{eqnarray*}
by a concavity argument: since we have a function concave in $m$, it takes values always greater than in its minima
which are on the domain boundary, namely in $m=0$ (for which we have $u_0[(d-\Lc_0)\beta+(t-l_0-u_1)\beta']$) and in
$m=u_0$ (for which we have $u_0\alpha$). Thus the function is always greater than in the value it takes at the smallest of its minima.

Analogously, for the second group $u_1$, taking into account that $\xin$ might contact among the live nodes
those who joined in the first group of $u_0$ nodes, we get
\begin{eqnarray*}
c_1(m) & \geq & m\alpha+(u_1-m)[(d-\Lc_0-u_0)\beta+ \\
       &      & (t-u_1+m-l_1)\beta']\\
       & \geq & u_1 \min\{\alpha,(d-\Lc_0-u_0)\beta+ (t-l_1-u_1)\beta'\}.
\end{eqnarray*}
By iteration and by summing over all the groups $u_0,\ldots,u_{g-1}$ such that $u_0+\ldots+u_{g-1}=k$ we get

\begin{eqnarray}
\sum_{i=0}^{g-1} u_i\min\{\alpha,(d-\Lc_0-\sum_{j=0}^{i-1}u_j)\beta+(t-l_i-u_i)\beta' \}.
\label{eq:selfishtradeoff}
\end{eqnarray}

\end{proof}

As explained in Section \ref{sec:RGC}, one of the two extremes in the storage-bandwidth trade-off is the minimum storage regeneration (MSR) point, which corresponds to the minimum amount of storage that is needed at each node to support data reconstruction by data collector by contacting $k$ nodes. The minimum storage point continues to be $\alpha=B/k$ under our selfishness model.

Since Proposition \ref{prop:selfishcollabRGC} is true for all possible values of $u_i$, it also holds particularly when $u_i=t-l_i$ for all $i$. Such a choice of $u_i$s eliminates the $\beta'$ component from the min-cut equation, allowing us to bound the value of $\beta$ at the MSR point as follows.

Recall that $\sum_{i=0}^{g-1}u_i=k$, hence, $gt-\Lc=k$, so when $u_i=t-l_i$ we have
\[
g=\frac{k+\Lc}{t}.
\]
For data reconstruction, we need $B \leq C_{r,s}(\alpha,\beta,\beta')$, hence
\[
B \leq \sum_{i=0}^{\frac{k+\Lc}{t}-1}(t-l_i)\min\{\alpha,(d-\Lc_0-\sum_{j=0}^{i-1}(t-l_j))\beta\},
\]
where $\alpha=B/k$.
Note that the expression on the right hand side is less than or equal to $\frac{B}{k}\sum_{i=0}^{\frac{k+\Lc}{t}-1}(t-l_i)$, which is however equal to $B$ (the same as the expression on the left hand side).

Thus, for every $i$
\[
(d-\Lc_0-\sum_{j=0}^{i-1}(t-l_j))\beta \geq B/k.
\]
Indeed, we know that having all the $min$ terms equal to $B/k$ gives $B$, thus it cannot be that one of the terms is strictly
smaller than $B/k$.
The expression on the left hand side is the smallest when $i=\frac{k+\Lc}{t}-1$, which in turn means
\[
(d-\Lc_0-\sum_{j=0}^{\frac{k+\Lc}{t}-2}(t-l_j))\beta \geq B/k.
\]
Consequently, the smallest feasible value for $\beta$ (which in turn leads to the smallest usage of bandwidth for regeneration) is
\begin{equation}\label{eq:betaMSRgeneral}
\frac{B/k}{(d-\Lc_0)-k+(t-l_{(k+\Lc)/t-1})}.
\end{equation}
This suggests that the bandwidth needed for download from the live nodes only depends on the last phase of regeneration,
where $d-\Lc_0$ and instead $d$ where contacted, and likewise, only $(t-l_{(k+\Lc)/t-1})$ nodes instead of $t-1$ actually participated in the collaborative phase.
We can thus conclude that
\begin{equation}\label{eq:betaMSR}
\frac{B/k}{(d-\Lc_0)-k+t} \leq \beta \leq \frac{B/k}{(d-\Lc_0)-k+(t-\lmax)}.
\end{equation}

We will like to specifically emphasize that the above bounds on $\beta$ are not to be confused with the range of values $\beta$ can take on the trade-off curve. Instead, what this result implies is that, even for the minimum storage point, the minimum feasible $\beta$ can be anywhere within this range, and depends on the precise number of selfish nodes involved in the collaborative phase, as noted in (\ref{eq:betaMSRgeneral}).

To compute $\beta'$, we consider the other extreme regime, where $u_i=1$ for all $i$, and thus $g=k$ (recall that we still
have $\alpha=B/k$). Then
\[
B \leq \sum_{i=0}^{k-1}\min\{B/k,(d-\Lc_0-i)\beta+(t-l_i-1)\beta'\}.
\]
Similarly as the computations done for $\beta$, since
\[
\min\{B/k,(d-\Lc_0-i)\beta+(t-l_i-1)\beta'\}\leq B/k
\]
we have that
\[
\sum_{i=0}^{k-1}\min\{B/k,(d-\Lc_0-i)\beta+(t-l_i-1)\beta'\}\leq B,
\]
and thus equality holds:
\[
\sum_{i=0}^{k-1}\min\{B/k,(d-\Lc_0-i)\beta+(t-l_i-1)\beta'\} = B.
\]
We observe that this is a sum of $k$ terms, so if any of the \emph{min} terms were smaller than $B/k$, there would be a contradiction. Thus, it must be that $(d-\Lc_0-i)\beta+(t-l_i-1)\beta' \geq B/k$ for all $i=0,...,k-1$. The smallest feasible $\beta'$ then corresponds to $i=k-1$, and we obtain that
\begin{equation}\label{eq:betaMSRprimeconstraint}
(d-\Lc_0-(k-1))\beta+(t-l_{k-1}-1))\beta'=B/k.
\end{equation}
This simplifies to
\[
\beta'= \frac{B/k-(d-\Lc_0-(k-1))\beta}{t-l_{k-1}-1}.
\]
Using (\ref{eq:betaMSR}), we determine that
\begin{equation}\label{eq:betaprimeMSRlb}
\frac{(B/k)(t-\lmax-1)}{(d-\Lc_0-k+t-\lmax)(t-1)} \leq \beta'
\end{equation}
and
\begin{equation}\label{eq:betaprimeMSRub}
\beta'
\leq \frac{(B/k)(t-1)}{(d-\Lc_0-k+t)(t-\lmax-1)}.
\end{equation}

With suitable choices of parameters $\Lc_0,\lmax$, the results from Proposition \ref{prop:rgc} on standard (independent) regenerations and Proposition \ref{prop:collabrgc} corresponding to collaborative regeneration can be deduced (not surprisingly) from the results of our generalization.

\begin{itemize}
\item
If $\lmax=0$, then there is no selfish node in the collaborative phase, only $\Lc_0$ live nodes might be selfish,
and thus the bounds described in (\ref{eq:betaMSR}) and (\ref{eq:betaprimeMSRlb})-(\ref{eq:betaprimeMSRub}) give
\[
\alpha=\frac{B}{k},~\beta=\beta'=\frac{B}{k}\frac{1}{d-\Lc_0-k+t}.
\]
Note that this is analogous to using less (i.e., $d-\Lc_0$ instead of $d$) nodes from among the live nodes for regeneration, and the specific result from Proposition \ref{prop:collabrgc} can be obtained by furthermore setting $\Lc_0=0$.
\item
If $\lmax$ takes its maximum value, that is $\lmax=t-1$, that would imply that there is no collaboration. The upper bound in (\ref{eq:betaMSR}) is then satisfied only corresponding to $t=1$, giving $\beta=\frac{B/k}{(d-\Lc_0)-k+1}$ which is analogous to the result from Proposition \ref{prop:rgc} for standard independent regeneration when $\Lc_0=0$. Also, $\lmax=t-1$ implies that the coefficient of $\beta'$ in (\ref{eq:betaMSRprimeconstraint}) is zero, and hence there is no information from the collaborative flows, and thus there is no practical meaning in discussing about $\beta'$.
\end{itemize}
These extreme cases are essentially a sanity check of our generalization, and the drawn conclusions are on expected lines. Similar conclusions can also be drawn about the other extreme point (minimum bandwidth regeneration) in the trade-off curve. Unlike the extreme points however, the intermediate points in the trade-off curve are not as amenable to closed form analysis, and comprise of an interesting regime, which we study using numerical optimization and discuss later in Section \ref{subsec:interp}.

\begin{figure*}[htbp]
    \begin{center}
    \subfigure[\label{fig:CollabSelRGCTradeoff}Trade-off under selfish attacks: $\Lc_0=1, l_i\leq1, \Lc=16/32$]{\includegraphics[scale=0.6]{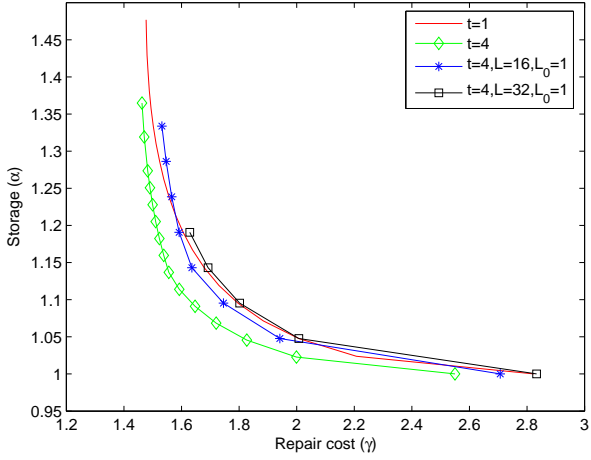}}\hspace*{0.1cm}
    \subfigure[\label{fig:CollabPolRGCTradeoff}Trade-off under pollution attacks: $\Bc_0=1, b_i\leq1, \Bc=16/32$]{\includegraphics[scale=0.6]{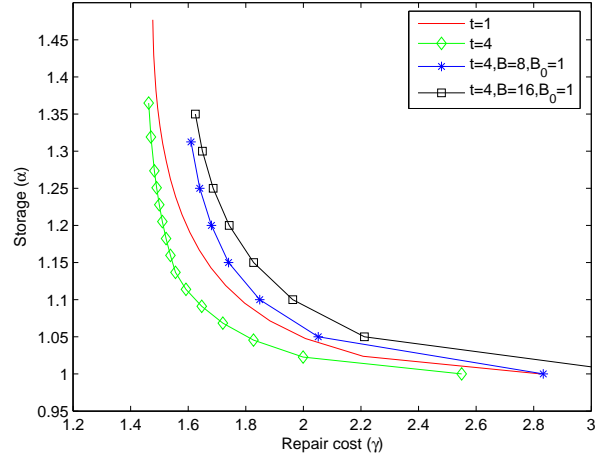}}
  \end{center}
  \caption{Storage-bandwidth tradeoff curves (normalized with $B/k$) using collaborative regenerating codes under Byzantine (selfish and pollution, respectively) attacks, determined by considering $g=32$ generations or regenerations where $t$ new nodes join and collaborate in each generation.}
  \label{fig:CollabRGCUnderAttacks}
  \vspace{-0.75cm}
\end{figure*}

\subsection{Effect of polluting nodes}

We now consider a worse case where the nodes are not selfish anymore, but are maliciously sending wrong data.
We assume that there are $\Bc_0$ polluting nodes among the live nodes in any generation of regeneration, while
$b_i\leq \bmax$ is the number of polluting nodes among the $i$th group of newcomers, with
$\Bc=\sum_{i=0}^{g-1}b_i$.

\begin{prop}
The resiliency capacity $C_{r,p}(\alpha,\beta,\beta')$ in the presence of polluting nodes is upper bounded by
\[
\begin{array}{c}
C_{r,p}(\alpha,\beta,\beta')\leq \\
\sum_{i=0}^g u_i\min\{\alpha,(d-2\Bc_0-\sum_{j=0}^{i-1}u_j)\beta+(t-2b_i-u_i)\beta' \}.
\end{array}
\]
\end{prop}

\begin{proof}
Let $u_0$ be the number of new comers contacted by the data collector in the first group of $t$ comers, with $m$ of them in $U$,
and $u_0-m$ of the others in $\bar{U}$. As in the proof above for selfish nodes, the contribution to the bound is $m\alpha$.

We now look at the $u_0-m$ nodes in $\bar{U}$. There are two contributions to the cut,
coming from either $\xin$ or $\xcor$. Take the first node. The $\xin$ part downloads from live nodes. Among these live nodes, there could be $\Bc_0$ polluting nodes. It thus gets a system of linear equations\footnote{Since all the network coding results used rely on linear network coding, we use an argument valid in this setting.} from the $d$ nodes, solving which would provide the unknown pieces of the encoded blocks. In the standard regeneration scenario, the unknowns correspond to the different pieces stored in the node itself. In the collaborative regeneration scenario, the unknowns include a subset of its own pieces, and additional information which allows it to collaborate and help other nodes regenerate.

There might or might not be wrong equations, depending on whether any of the live Byzantine nodes are contacted, but to be able to detect them, a naive, brute-force technique will be to solve all possible valid combinations (determined by the number of unknowns) of the subsets of equations, and choose the solution which concurs in majority of these combinations. Independently of even if more elegant mechanisms are employed, in order to actually figure out which equations are valid, it requires $\Bc_0$ good equations to compensate for the $\Bc_0$ potentially wrong ones. This is a more fundamental limit in Byzantine settings \cite{ByzantineGeneralProblem}. Having said that, we will like to note that if some extrinsic information is available, better Byzantine fault tolerance may be achievable, which we will briefly discuss later in Section \ref{sec:practicalconsiderations}. However, the rest of this section continues the analysis under the assumption that no other extrinsic (side-channel) information is available.

Thus among the $d$ nodes contacted, those which will provide actual information to recover the lost data contribute to the cut by only
$(d-2\Bc_0)\beta$. Now for $\xcor$, it contacts $t-1$ other new comers (together with the edge $\xin\rightarrow\xcor$), $u_0-m$ could already be in $\bar{U}$, and $b_0$ could be bad, thus using the same argument as for $\Bc_0$, the contribution to the cut is
$(t-(u_0-m)-2b_0)\beta'$, for a total of
\begin{eqnarray*}
c_0(m) & \geq & m\alpha+(u_0-m)[(d-2\Bc_0)\beta+  \\
       &      & (t-u_0+m-2b_0)\beta']\\
       & \geq & u_0 \min\{\alpha,(d-2\Bc_0)\beta+(t-2b_0-u_1)\beta'\}.
\end{eqnarray*}
Likewise, for the second group $u_1$, we get
\begin{eqnarray*}
c_1(m) & \geq & m\alpha+(u_1-m)[(d-2\Bc_0-u_0)\beta+ \\
       &      & (t-u_1+m-2b_2)\beta']\\
       & \geq & u_1 \min\{\alpha,(d-2\Bc_0-u_0)\beta+(t-2b_2-u_1)\beta'\}.
\end{eqnarray*}
By iterating and by summing over all the groups $u_0,\ldots,u_{g-1}$ such that $u_0+\ldots+u_{g-1}=k$ we get
\[
\sum_{i=0}^g u_i\min\{\alpha,(d-2\Bc_0-\sum_{j=0}^{i-1}u_j)\beta+(t-2b_i-u_i)\beta' \}.
\]
\end{proof}


\subsection{Interpretation of the analysis}\label{subsec:interp}

We are interested in understanding the effects of both selfish and polluting nodes on the storage-bandwidth storage
trade-off curve. To do so, we numerically minimize the bandwidth under the respective min-cut constraints, and report some of our results in
Fig \ref{fig:CollabRGCUnderAttacks} corresponding to $d=48$, $k=32$, $t=4$, $g=32$ and compare how
the trade-off curves for different adversarial scenarios behave with respect to both collaborative and standard regenerating codes. 

In Fig.~\ref{fig:CollabRGCUnderAttacks} (a), selfish nodes are introduced in the network. We fix their maximum number among the live nodes to be only $\Lc_0=1$, and similarly $\lmax=1$ bounds the number of selfish nodes during collaboration. We consider two cases: when $\Lc=16$, that is all together 16 selfish nodes interfered during collaboration, and $\Lc=32$, that is one selfish node was present at each stage of the regeneration process. The optimization was performed by letting the
parameters $\beta,\beta'$ range through a range of values limited by the MSR and MBR points. Derivation for the MSR points were provided above, analogous formulas can be derived for the MBR points. We observe in Fig.\ref{fig:CollabSelRGCTradeoff} that when only half of the $g$ groups had selfish nodes ($\Lc=16$), the performance gets close to standard regenerating codes for a middle
range of repair cost values, while it is even worse for $\Lc=32$. For the later, the trade-off curve is worse, as
expected, since not only the collaboration phase is not contributing, but there is furthermore one selfish node in the live
nodes themselves.

In Fig.~\ref{fig:CollabPolRGCTradeoff}, the same setting is repeated, this time with polluting nodes. We see that even a small number of pollutant nodes in a collaborative regeneration group, or among the live nodes leads to drastic deterioration of what can be achieved using collaborative regeneration - casting some doubt on the efficacy of regenerating codes. In practice, some additional extrinsic mechanisms can alleviate the situation, which we will briefly mention in Section \ref{sec:practicalconsiderations}.

It is important to note that the plot for pollution attacks corresponds to the case where polluting nodes actually answer correctly to the
request of a data collector, meaning in particular that the minimum storage point is still $\alpha=B/k$. If it were not
the case, namely, the polluting nodes could give wrong data to the data collector, then the minimum storage point would
shift to $\alpha=B/(k-2\Bc_0)$. Further analysis is needed to comprehend the impact of the same, which we defer for future investigation.

%
%

\section{Exact Collaborative Regenerating Codes}
\label{sec:ICCExactRGC}

Currently, \cite{Shum-ICC} is, up to our knowledge, the only example of explicit codes for exact regeneration with collaboration, which works specifically for only the minimum storage regeneration point. We will first recall the construction, before considering it in the context of Byzantine adversaries. Note that in presence of Byzantine nodes, the number of nodes to be accessed might be different than what is used if there are no Byzantine nodes, for example as noted above, the minimum storage point is shifted from $B/k$ to $B/(k-2\Bc_0)$ where $\Bc_0$ is the number of Byzantine nodes that might send wrong information during data collection. Thus in what follows, we will retain $k$ to denote the number of nodes that the data collector accesses to retrieve the data stored, while $\kappa$ is used as the dimension of the codes used, such as for Reed-Solomon codes.

Consider the $(n,\kappa)$ Reed-Solomon code which is defined
over the finite field $\FF_q$ with
$q \geq n$ a power of a prime. Suppose that the object $\ov$ to be stored in $n$ nodes
can be written as
$\ov^T=(\ov_{11},\ldots,\ov_{1\kappa},\ldots,\ov_{t1},\ldots,\ov_{t\kappa})$ with
$\ov_{ij}$ in either $\FF_q$ or any finite field extension of $\FF_q$.
Note that this means that the object is cut into a number of pieces which depends on the number $t$ of (predetermined, expected) failures,\footnote{While such an assumption is somewhat restrictive, and design of more adaptive codes constitute an interesting future direction of research, we note that such codes can nevertheless be practically used either by over-estimating the number of faults (though this may not be optimal anymore), and also when failures are corrected lazily by deliberately postponing the repair process till a predetermined number of faults are accumulated.} with $k<n-t$. Furthermore, \cite{Shum-ICC} considers only the regime $k=d$.

The generator matrix $G$ of the Reed-Solomon code is a $\kappa\times n$ Vandermonde matrix whose columns are denoted by $\gv_i$, $i=1,\ldots,n$.
Every node is assumed to  know $G$.
Now create a matrix ${\bf O}$ as follows:
\[
{\bf O}=\left[
\begin{array}{ccc}
\ov_{11} & \ldots & \ov_{1\kappa} \\
         &        &          \\
\ov_{t1} & \ldots, & \ov_{t\kappa}
\end{array}
\right].
\]
The $i$th node stores ${\bf O}\gv_i$ where $\gv_i$ denotes the $i$th column of $G$, for example,
node 1 stores
\[
{\bf O}\gv_1.
\]

The $t$ rows represent what we will call the $t$ \emph{pieces} that
the corresponding node stores. That is, the \emph{encoded data
block} stored by each node comprises of \emph{multiple pieces}.
We will use the size of such a piece to define one unit of
data.

Any choice of $\kappa$ nodes $i_1,\ldots,i_\kappa$ clearly allows to
retrieve $\ov$ since we get
\[
{\bf O}[\gv_{i_1},\ldots,\gv_{i_\kappa}]
\]
where the matrix formed by any $\kappa$ columns of $G$ is a Vandermonde matrix and is thus invertible.

Let us now assume that $t$ nodes go offline, and $t$ new nodes join. Let
us call the $t$ new nodes as nodes 1 to $t$. The $i$th newcomer
will ask $(\ov_{i1},\ldots,\ov_{i\kappa})\gv_j$ for any choice of $\kappa$
nodes among the live nodes.

Each newcomer can invert the matrix formed by the columns of $G$,
and each decode $(\ov_{i1},\ldots,\ov_{i\kappa})$ respectively. Thus it can compute the
piece corresponding to its own first row, and also can compute
$(\ov_{i1},\ldots,\ov_{i\kappa})\gv_j$ and send it to the $j$th node,
which all will do similar computations and likewise deliver the missing pieces to the other newcomers, hence completing the collaborative regeneration process.

\begin{table*}
\center
\begin{minipage}{3in}
\begin{tabular}{c|c|c|c}
cost     & $l_0=\Lc_0=0$ & $l_0=1,\Lc_0=0$ & $l_0=0,\Lc_0=1$ \\
\hline
$\beta$  & 1/2    & 1       & $\beta_{av}$=3/4 \\
$\beta'$ & 1/2    & 0       & 1/2 \\
$\gamma$ & 2      & 3       & 2
\end{tabular}
\caption{selfish nodes: $\alpha=1$, $t=2$, $d=3$}\label{tab:sel-cost}
\end{minipage}
\hspace{1.5cm}
\begin{minipage}{3in}
\begin{tabular}{c|c|c|c}
cost     & $b_0=\Bc_0=0$ & $b_0=1,\Bc_0=0$ & $b_0=0,\Bc_0=1$ \\
\hline
$\beta$  & 1/2    & 1       & 1/2 \\
$\beta'$ & 1/2    & 0       & 1/2 \\
$\gamma$ & 2      & 3       & 3 ($d=5$)\\
\end{tabular}
\caption{polluting nodes: $\alpha=1$, $t=2$, $d=3$} \label{tab:pol-cost}
\end{minipage}
\vspace{-1cm}
\end{table*}

\begin{ex}\label{ex:icc}\rm
Consider the $(n,\kappa)=(7,3)$ Reed-Solomon code which is defined
over the finite field $\FF_8=\{0,1,w,w^2,w^3,w^4,w^5,w^6,w^7\}$ with
$w^3=w+1$. Suppose that the object $\ov$ is to be stored in $n=7$ nodes,
while expecting to deal with $t=2$ failures. First, represent the object as
$\ov^T=(\ov_{11},\ov_{12},\ov_{13},\ov_{21},\ov_{22},\ov_{23})$ with
$\ov_{ij}$ in either $\FF_8$ or any finite field extension of
$\FF_8$, say $\FF_q$. The generator matrix $G$ of the Reed-Solomon
code is given by:
\[
{\footnotesize \left[
\begin{array}{ccccccc}
1 & 1 & 1 & 1 & 1 & 1 & 1 \\
w & w^2 & 1+w & w+w^2 & 1+w+w^2 & 1+w^2 & 1 \\
w^2 & w+w^2 & 1+w^2 & w & 1+w & 1+w+w^2 & 1 \\
\end{array}
\right]}.
\]
Now create a matrix ${\bf O}$ as follows:
\[
{\bf O}=\left[
\begin{array}{ccc}
\ov_{11} & \ov_{12} & \ov_{13} \\
\ov_{21} & \ov_{22}, & \ov_{23}
\end{array}
\right].
\]
The $i$th node stores ${\bf O}\gv_i$ where $\gv_i$ denotes the $i$th column of $G$, for example,
node 1 stores
\[
{\bf O}\left[
\begin{array}{c}
1\\
w\\
w^2
\end{array}
\right]=
\left[
\begin{array}{ccc}
\ov_{11}+\ov_{12} w+\ov_{13}w^2.\\
\ov_{21}+\ov_{22}w+\ov_{23}w^2
\end{array}
\right].
\]

Thus each encoded data block comprises of two pieces of size one unit each in this
example, and the original object is of size six units,
and each encoded block is of size two units.

Any choice of $k=\kappa=3$ nodes $i_1,i_2,i_3$ clearly allows to
retrieve $\ov$ since we get
\[
{\bf O}[\gv_{i_1},\gv_{i_2},\gv_{i_3}]
\]
where the matrix formed by any 3 columns of $G$ is a Vandermonde matrix and is thus invertible.

Let us now assume that 2 nodes go offline, and 2 new nodes join. Let
us call the two new nodes as node 1 and node 2. The first new comer
will ask $(\ov_{11},\ov_{12},\ov_{13})\gv_i$ for any choice of 3
nodes among the 5 live nodes, while the second new comer will
similarly ask $(\ov_{21},\ov_{22},\ov_{23})\gv_i$ from any of the 5
live nodes.

Both new comers can invert the matrix formed by the columns of $G$,
and decode each respectively $(\ov_{11},\ov_{12},\ov_{13})$ and
$(\ov_{21},\ov_{22},\ov_{23})$. Now the first node can compute the
piece corresponding its own first row, and also can compute
$(\ov_{11},\ov_{12},\ov_{13})\gv_2$ and send it to the second node,
which likewise can compute $(\ov_{21},\ov_{22},\ov_{23})\gv_1$ and
send it to node 1, which completes the regeneration process.

Thus, overall, eight units of data transfer is needed in this
example, in order to replenish four units of lost data.
Note that if one node did regeneration of two pieces using six data transfer, it could
send the other node the other two pieces directly, needing again a
total of eight units of data transfer. As mentioned previously in Section \ref{sec:RGC}, when $d=k$
as is the case for this code construction, the repair cost is the same as that of erasure codes, though
with a better load balancing, as seen in this example.
\end{ex}

We use the above toy example to illustrate the effect of selfish/polluting nodes. We consider two scenarios with selfish nodes: (i) Consider that one of the two newcomer nodes does not agree to collaborate with the other. In this case, the other node has no choice than to download more data from the live nodes. Given that each node contacts $d=3$ live nodes, this means downloading 2 encoded pieces from each of the 3 nodes, for a total of 6 pieces. The cost of one repair is then $\gamma=3$. Note that all the bandwidth costs in this example are normalized with $B/k=2$. Note also that, in a general scenario, during the collaborative regeneration process, different nodes may face different number of selfish nodes, affecting accordingly the necessary bandwidth for regenerations. (ii) Consider now that both newcomer nodes collaborate, but there is $\Lc_0=1$ selfish node among the live nodes. In the worst case, both newcomers try two well behaved lives nodes and the same non-responding node. It might or not be easy for these newcomers to contact other nodes that are willing to help with the download. So if the newcomers decide to keep on downloading more from only the already responding nodes, we get that they each need to download at least 1 piece of data from one responding live node, and 2 pieces from the other responding live node, for an average of $\beta_{av}=(1+1/2)/2=3/4$ download bandwidth, after which collaboration can proceed as normal. The bandwidth costs are summarized in Table \ref{tab:sel-cost}. It can be seen that in this case $\Lc_0$ is not harmful for total bandwidth cost per repair, though it does imbalance the network load.

Let us now consider the case of polluting nodes, where we first assume that the polluting nodes do not interfere with data collection. (i) If one of the two collaborating nodes is polluting, then the other node has no choice but to retrieve the whole object from the live nodes, and no collaboration is possible. In this particular toy example, this gives the same end result as with one selfish collaborating node, since in both cases reconstructing the object is needed. (ii) If $\Bc_0=1$, in the worst case, both collaborating nodes get 1 fake encoded piece of data, and 2 genuine ones. Now to check which data, if any, is corrupted, 2 more genuine encoded fragments are needed. However, since the nodes do not know which of the live nodes might have gone rogue, they are forced to contact the remaining two more nodes. This inflates the number $d=3$ to $d=5$, the maximum amount of available live nodes here. These results are summarized in Table \ref{tab:pol-cost}.

Finally, in the worst case, the polluting nodes can also send wrong information to the data collector.
Since the stored data at the live nodes is encoded using Reed-Solomon code in this example, it is resistant to errors, as long as the number of errors is not more than twice the maximal number of tolerated erasures. However, this also means that either the number of contacted nodes is increased, or for a fixed $k$, the amount of data stored in each node has to be increased.

In this example, since we have 5 live nodes, only one polluting node sending wrong information to the data collector can be tolerated.
More generally, a $(n,\kappa)$ Reed-Solomon code is known to tolerate $n_s=n-\kappa$ erasures, or $n_b=(n-\kappa)/2$ errors, or more
generally $n_s$ erasures and $n_b$ errors as long as $n_s+2n_b \leq n-\kappa$.

%
%

%
%

\section{Practical considerations}
\label{sec:practicalconsiderations}
In practice, the number of Byzantine nodes is not known a priori. While selfish nodes are trivially dealt with, pollutants can not be detected a priori, and hence are difficult to deal with. Thus, regenerating nodes may try to first regenerate with responses from the minimal number of nodes, assuming (possibly, wrongly) that there are no pollutants. If there are however pollutants, then the regenerated block will be different from what ought to have been regenerated. For exact regeneration, a globally known hash function, and prior, secure and globally accessible look-up table with the hashes (signature) for the encoded fragments of an object can be used, to verify with low communication overhead whether the regenerated block is correct or not. If integrity violation is detected, then progressively more nodes data may be downloaded, possibly by contacting more nodes. Such an extrinsic information can alleviate the effect of Byzantine nodes. As soon as the node has enough good information to regenerate, it can be easily verified, thus, there is no need to waste one bit of good information just to negate each wrong bit. For the example in Section \ref{sec:ICCExactRGC}, with the use of such extra information, if there is one pollutant among the live nodes, the regenerations could be carried out by contacting at most four of the live nodes, and one can also tolerate upto two Byzantine nodes - both infeasible without such extra information.

The actual achieved system performance will depend on the precise protocol details, and there will in all cases be additional protocol overheads, both in terms of storage as well as bandwidth needs. Such systems considerations were beyond the scope of the current paper which studies the theoretical constraints of regenerating codes in the presence of Byzantine nodes. Furthermore, regenerating codes incur high computational complexity \cite{biersackRGCstudy} even without consideration of Byzantine failures. Byzantine nodes will further amplify the computational overheads. Thus, even though regenerating codes have promising qualities (theoretically), and have been much studied in the last few years, all these practical issues need to be taken into account and studied holistically, to determine their benefits and trade-offs in practice.


%
%

\section{Related works}
We have already provided a concise survey of regenerating codes related literature in the discussion precursing the new bounds for collaborative regenerating codes under Byzantine faults determined in this paper. Thus, here we will discuss about pollution attacks in general, both in the context of different kinds of peer-to-peer systems, and in the context of network coding.

Pollution attacks are mitigated in peer-to-peer content dissemination systems \cite{IntegrityP2PStreaming,p2ptvpollution,p2pRSS} using a combination of proactive strategies such as digital signature provided by the content source or by reactive strategies such as by randomized probing of the content source, leveraging on the causal relationship in the sequence of content to be delivered, as well as by deploying reputation mechanisms. In such settings, the prevention of pollution attacks is furthermore facilitated by a continuous involvement of the content source, which is assumed to be online.

Generally speaking, P2P storage environments are fundamentally different from P2P content distribution networks. The content owner may or not be online all the while. Furthermore, the very premise of regenerating codes is a setting where no one node possesses the whole copy of the object to be stored, i.e., a hybrid storage strategy where one full copy of the data is stored in addition to the encoded blocks, is excluded for other practical considerations. Likewise, different stored objects may be independent of each other. Hence, mechanisms to provide protection against errors as an inherent property of the code (similar to error correcting codes) becomes essential. The presented study looks at the fundamental capacity of such codes under some specific adversarial models. This work is thus complementary to other existing storage systems approaches such as incentive and reputation mechanisms \cite{samsara} and remote data checking techniques \cite{RDC-2010} for data outsourced to third parties to name a few. Likewise, Byzantine algorithms have been used in Oceanstore \cite{oceanstore} to support reliable data updates. The focus there is on application level support for updating content, rather than storage infrastructure level Byzantine behavior studied in this paper.

Pollution attacks have also been studied specifically in the context of network coding where it has already been noticed that though collaboration among the nodes through coding does increase the throughput, it also makes the network much more vulnerable to pollution attacks than under traditional routing. To remedy this threat, several authentication techniques have been studied in the context of network coding, such as digital signatures (for e.g., \cite{CJL,YWRG,ZKMH,BFKW} and authentication codes \cite{OF}.

%
%

\section{Conclusion}

Leveraging on network coding results, regenerating codes were introduced as a redundancy technique in networked distributed storage. Collaboration among the nodes participating in the regeneration process has recently been shown to improve the storage-bandwidth trade-offs. In this paper we determine the resilience capacity of collaborative regeneration in the presence of selfish or polluting nodes, and expose that collaboration may be detrimental under Byzantine attacks to such an extent that it may instead be better not to collaborate. We also show that, while collaborative regeneration is extremely vulnerable as a stand alone process, Byzantine attacks can be easily mitigated using some additional extrinsic information.

%
%

\end{document}